\documentclass[letterpaper,USenglish,numberwithinsect]{article}

\title{Algorithmic information, plane Kakeya sets, and conditional dimension}

\author{
	Jack H. Lutz\footnote{Research supported in part by National Science Foundation Grants 1247051 and 1545028.}\\
	Department of Computer Science, Iowa State University\\
	Ames, IA 50011, USA\\
	\texttt{lutz@cs.iastate.edu}
	\and
	Neil Lutz\footnote{This work was conducted at DIMACS and at Hebrew University. It was partially enabled through support from the National Science Foundation under grants CCF-1445755 and CCF-1101690.}\\
	Department of Computer Science, Rutgers University\\
	Piscataway, NJ 08854, USA\\
	\texttt{njlutz@rutgers.edu}}

\usepackage{amsmath}
\usepackage{amssymb}
\usepackage{amsthm}
\usepackage{enumitem}
\setlist{noitemsep}
\usepackage{mathtools}

\newcommand\numberthis{\addtocounter{equation}{1}\tag{\theequation}}
\numberwithin{equation}{section}

\theoremstyle{plain}
\newtheorem{thm}{Theorem}
\newtheorem{obs}[thm]{Observation}

\newtheorem{lem}[thm]{Lemma}

\theoremstyle{remark}
\newtheorem{cla}{Claim}

\DeclareMathOperator{\Dim}{Dim}
\DeclareMathOperator{\mdim}{mdim}
\DeclareMathOperator{\Mdim}{Mdim}

\newcommand{\R}{\mathbb{R}}
\newcommand{\Z}{\mathbb{Z}}
\newcommand{\N}{\mathbb{N}}
\newcommand{\Q}{\mathbb{Q}}

\newcommand{\ve}{\varepsilon}

\newenvironment{proofof}[1]{\begin{trivlist}
		\item[\hskip \labelsep \textit{Proof of #1.}]}{\end{trivlist}}
	
\begin{document}
		\maketitle
		\thispagestyle{empty}
		\begin{abstract}
			We formulate the \emph{conditional Kolmogorov complexity} of $x$ \emph{given} $y$ at \emph{precision} $r$, where $x$ and $y$ are points in Euclidean spaces and $r$ is a natural number.  We demonstrate the utility of this notion in two ways.
			\begin{enumerate}
				\item
					We prove a \emph{point-to-set principle} that enables one to use the (relativized, constructive) dimension of a \emph{single point} in a set $E$ in a Euclidean space to establish a lower bound on the (classical) Hausdorff dimension of $E$.  We then use this principle, together with conditional Kolmogorov complexity in Euclidean spaces, to give a new proof of the known, two-dimensional case of the Kakeya conjecture.  This theorem of geometric measure theory, proved by Davies in 1971, says that every plane set containing a unit line segment in every direction has Hausdorff dimension $2$.
				\item
					We use conditional Kolmogorov complexity in Euclidean spaces to develop the \emph{lower} and \emph{upper conditional dimensions} $\dim(x|y)$ and $\Dim(x|y)$ of $x$ given $y$, where $x$ and $y$ are points in Euclidean spaces.  Intuitively these are the lower and upper asymptotic algorithmic information densities of $x$ conditioned on the information in $y$.  We prove that these conditional dimensions are robust and that they have the correct information-theoretic relationships with the well-studied dimensions $\dim(x)$ and $\Dim(x)$ and the mutual dimensions $\mdim(x:y)$ and $\Mdim(x:y)$.
			\end{enumerate}
		\end{abstract}
		\clearpage
		\setcounter{page}{1}
		\section{Introduction}
		
		This paper concerns the fine-scale geometry of algorithmic information in Euclidean spaces. It shows how new ideas in algorithmic information theory can shed new light on old problems in geometric measure theory. This introduction explains these new ideas, a general principle for applying these ideas to classical problems, and an example of such an application. It also describes a newer concept in algorithmic information theory that arises naturally from this work.
		
		Roughly fifteen years after the mid-twentieth century development of the \emph{Shannon information theory} of probability spaces~\cite{Shan48}, Kolmogorov recognized that Turing's mathematical theory of computation could be used to refine the Shannon theory to enable the amount of information in individual data objects to be quantified~\cite{Kolm65}.  The resulting theory of \emph{Kolmogorov complexity}, or \emph{algorithmic information theory}, is now a large enterprise with many applications in computer science, mathematics, and other sciences~\cite{LiVit08}.  Kolmogorov proved the first version of the fundamental relationship between the Shannon and algorithmic theories of information in~\cite{Kolm65}, and this relationship was made exquisitely precise by Levin's coding theorem~\cite{Levi73,Levi74}.  (Solomonoff and Chaitin independently developed Kolmogorov complexity at around the same time as Kolmogorov with somewhat different motivations~\cite{Solo64,Chai66,Chai69}.)
		
		At the turn of the present century, the first author recognized that Hausdorff's 1919 theory of fractal dimension~\cite{Haus19} is an older theory of information that can also be refined using Turing's mathematical theory of computation, thereby enabling the \emph{density} of information in individual infinite data objects, such as infinite binary sequences or points in Euclidean spaces, to be quantified~\cite{Lutz03a,Lutz03b}.  The resulting theory of \emph{effective fractal dimensions} is now an active enterprise with a growing array of applications~\cite{DowHir10}.  The paper~\cite{Lutz03b} proved a relationship between effective fractal dimensions and Kolmogorov complexity that is as precise as --- and uses --- Levin's coding theorem.

		Most of the work on effective fractal dimensions to date has concerned the \emph{(constructive) dimension} $\dim(x)$ and the dual \emph{strong (constructive) dimension} $\Dim(x)$~\cite{AHLM07} of an infinite data object $x$, which for purposes of the present paper is a point in a Euclidean space $\R^n$ for some positive integer $n$.\footnote{These constructive dimensions are $\Sigma^0_1$ effectivizations of Hausdorff and packing dimensions~\cite{Falc14}. Other effectivizations, e.g., computable dimensions, polynomial time dimensions, and finite-state dimensions, have been investigated, but only the constructive dimensions are discussed here.}
		The inequalities 
		\[0\leq\dim(x)\leq \Dim(x)\leq n\]
		hold generally, with, for example,  $\Dim(x) = 0$ for points $x$ that are computable and $\dim(x) = n$ for points that are algorithmically random in the sense of Martin-L\"{o}f [Mart66].
		
		How can the dimensions of individual points---dimensions that are defined using the theory of computing---have any bearing on classical problems of geometric measure theory? The problems that we have in mind here are problems in which one seeks to establish lower bounds on the classical Hausdorff dimensions $\dim_H(E)$ (or other fractal dimensions) of sets $E$ in Euclidean spaces. Such problems involve global properties of sets and make no mention of algorithms.
		
		The key to bridging this gap is relativization. Specifically, we prove here a \emph{point-to-set principle} saying that, in order to prove a lower bound $\dim_H(E)\geq \alpha$, it 
		suffices to show that, for every $A \subseteq \N$ and every $\ve>0$, there is a point $x \in E$ such that $\dim^A(x)\geq\alpha-\ve$, where $\dim^A(x)$ is the dimension of $x$ relative to the oracle $A$. We also prove the analogous point-to-set principle for the classical packing dimension $\dim_P(E)$ and the relativized strong dimension $\Dim^A(x)$.
		
		We illustrate the power of the point-to-set principle by using it to give a new proof of a known theorem in geometric measure theory. A Kakeya set in a Euclidean space $\R^n$ is a set $K \subseteq \R^n$ that contains a unit line segment in every direction. Besicovitch~\cite{Besi19,Besi28b} proved that Kakeya sets can have Lebesgue measure $0$ and asked whether Kakeya sets in the Euclidean plane can have dimension less than $2$~\cite{Davi71}. The famous Kakeya conjecture asserts a negative answer to this and to the analogous question in higher dimensions, i.e., states that every Kakeya set in a Euclidean space $\R^n$ has Hausdorff dimension $n$.\footnote{Statements of the Kakeya conjecture vary in the literature. For example, the set is sometimes required to be compact or Borel, and the dimension used may be Minkowski instead of Hausdorff. Since the Hausdorff dimension of a set is never greater than its Minkowski dimension, our formulation is at least as strong as those variations.} This conjecture holds trivially for $n=1$ and was proven by Davies~\cite{Davi71} for $n=2$. A version of the conjecture in finite fields has been proven by Dvir~\cite{Dvir09}. For Euclidean spaces of dimension $n\geq 3$, it is an important open problem with deep connections to other problems in analysis~\cite{Wolf99,Tao00}.
		
		In this paper we use our point-to-set principle to give a new proof of Davies's theorem. This proof does not resemble the classical proof, which is not difficult but relies on Marstrand's projection theorem~\cite{Mars54} and point-line duality. Instead of analyzing the set $K$ globally, our proof focuses on the information content of a single, judiciously chosen point in $K$. Given a Kakeya set $K\subseteq\R^2$ and an oracle $A\subseteq\N$, we first choose a particular line segment $L\subseteq K$ and a particular point $(x,mx+b)\in L$, where $y=mx+b$ is the equation of the line containing $L$.\footnote{One might na\"{i}vely expect that for independently random $m$ and $x$, the point $(x,mx+b)$ must be random. In fact, in every direction there is a line that contains no random point~\cite{LutLut15a}.} We then show that $\dim^A(x,mx+b)\geq2$. By our point-to-set principle this implies that $\dim_H(K)\geq2$.
		
		Our proof that $\dim^A(x,mx+b)\geq2$ requires us to formulate a concept of conditional Kolmogorov complexity in Euclidean spaces. Specifically, for points $x\in\R^m$ and $y\in\R^n$ and natural numbers $r$, we develop the \emph{conditional Kolmogorov complexity} $K_r(x|y)$ of $x$ \emph{given} $y$ at \emph{precision} $r$. This is a ``conditional version'' of the Kolmogorov complexity $K_r(x)$ of $x$ at precision $r$ that has been used in several recent papers (e.g., \cite{LutMay08,CasLut15,GLMM14}).
		
		In addition to enabling our new proof of Davies's theorem, conditional Kolmogorov complexity in Euclidean spaces enables us to fill a gap in effective dimension theory. The fundamental quantities in Shannon information theory are the \emph{entropy} (information content) $H(X)$ of a probability space $X$, the \emph{conditional entropy} $H(X|Y)$ of a probability space $X$ given a probability space $Y$, and the \emph{mutual information} (shared information) $I(X;Y)$ between two probability spaces $X$ and $Y$~\cite{CovTho06}.  The analogous quantities in Kolmogorov complexity theory are the \emph{Kolmogorov complexity} $K(u)$ of a finite data object $u$, the \emph{conditional Kolmogorov complexity} $K(u|v)$ of a finite data object $u$ given a finite data object $v$, and the \emph{algorithmic mutual information} $I(u:v)$ between two finite data objects $u$ and $v$~\cite{LiVit08}. The above-described dimensions $\dim(x)$ and $\Dim(x)$ of a point $x$ in Euclidean space (or an infinite sequence $x$ over a finite alphabet)
		are analogous by limit theorems~\cite{Mayo02,AHLM07} to $K(u)$ and hence to $H(X)$. Case and the first author have recently developed and investigated the \emph{mutual dimension} $\mdim(x:y)$ and the dual \emph{strong mutual dimension} $\Mdim(x:y)$, which are densities of the algorithmic information shared by points $x$ and $y$ in Euclidean spaces~\cite{CasLut15} or sequences $x$ and $y$ over a finite alphabet~\cite{CasLut15a}.  These mutual dimensions are analogous to $I(u:v)$ and $I(X;Y)$.
				
		What is conspicuously missing from the above account is a notion of conditional dimension. In this paper we remedy this by using conditional Kolmogorov complexity in Euclidean space to develop the \emph{conditional dimension} $\dim(x|y)$ of $x$ given $y$ and its dual, the \emph{conditional strong dimension} $\Dim(x|y)$ of $x$ given $y$, where $x$ and $y$ are points in Euclidean spaces. We prove that these conditional dimensions are well behaved and that they have the correct information theoretic relationships with the previously defined dimensions and mutual dimensions. The original plan of our proof of Davies's theorem used conditional dimensions, and we developed their basic theory to that end.  Our final proof of Davies's theorem does not use them, but conditional dimensions (like the conditional entropy and conditional Kolmogorov complexity that motivate them) are very likely to be useful in future investigations.
		
		The rest of this paper is organized as follows. Section~\ref{sec:dpes} briefly reviews the dimensions of points in Euclidean spaces. Section~\ref{sec:ps} presents the point-to-set principles that enable us to use dimensions of individual points to prove lower bounds on classical fractal dimensions. Section~\ref{sec:ckces} develops conditional Kolmogorov complexity in Euclidean spaces. Section~\ref{sec:ksp} uses the preceding two sections to give our new proof of Davies's theorem. Section~\ref{sec:cdes} uses Section~\ref{sec:ckces} to develop conditional dimensions in Euclidean spaces. Most proofs are deferred to the optional technical appendix.
		
		\section{Dimensions of Points in Euclidean Spaces}\label{sec:dpes}
		
		This section reviews the constructive notions of dimension and mutual dimension in Euclidean spaces. The presentation here is in terms of Kolmogorov complexity. Briefly, the \emph{conditional Kolmogorov complexity} $K(w|v)$ of a string $w \in \{0,1\}^*$ \emph{given} a string $v\in\{0,1\}^*$ is the minimum length $|\pi|$ of a binary string $\pi$ for which $U(\pi,v) = w$, where $U$ is a fixed universal self-delimiting Turing machine. The \emph{Kolmogorov complexity} of $w$ is $K(w|\lambda)$, where $\lambda$ is the empty string. We write $U(\pi)$ for $U(\pi,\lambda)$. When $U(\pi) = w$, the string $\pi$ is called a \emph{program} for $w$. The quantity $K(w)$ is also called the \emph{algorithmic information content} of $w$. Routine coding extends this definition from $\{0,1\}^*$ to other discrete domains, so that the Kolmogorov complexities of natural numbers, rational numbers, tuples of these, etc., are well defined up to additive constants. Detailed discussions of self-delimiting Turing machines and Kolmogorov complexity appear in the books~\cite{LiVit08,Nies09,DowHir10} and many papers.
		
		The definition of $K(q)$ for rational points $q$ in Euclidean space is lifted in two steps to define the dimensions of arbitrary points in Euclidean space. First, for $x \in \R^n$ and $r \in \N$, the \emph{Kolmogorov complexity} of $x$ at \emph{precision} $r$ is
		\begin{equation}\label{eq:Kr}
		K_r(x)=\min\{K(q)\,:\,q\in\Q^n\cap B_{2^{-r}}(x)\}\,,
		\end{equation}
		where $B_{2^{-r}}(x)$ is the open ball with radius $2^{-r}$ and center $x$. Second, for $x \in \R^n$, the \emph{dimension} and \emph{strong dimension} of $x$ are
		\begin{equation}\label{eq:dimDim}
		\dim(x)=\liminf_{r\to\infty}\frac{K_r(x)}{r}\qquad\textrm{and}\qquad\Dim(x)=\limsup_{r\to\infty}\frac{K_r(x)}{r}\,,
		\end{equation}
		respectively.\footnote{We note that $K_r(x)=K(x\upharpoonright r)+o(r)$, where $x\upharpoonright r$ is the binary expansion of $x$, truncated $r$ bits to the right of the binary point. However, it has been known since Turing's famous correction~\cite{Turi37} that binary notation is not a suitable representation for the arguments and values of computable functions on the reals. (See also~\cite{Weih00}.) Hence, in order to make our definitions useful for further work in computable analysis, we formulate complexities and dimensions in terms of rational approximations, both here and later.}
		
		Intuitively, $\dim(x)$ and $\Dim(x)$ are the lower and upper asymptotic densities of the algorithmic information in $x$. These quantities were first defined in Cantor spaces using betting strategies called gales and shown to be constructive versions of classical Hausdorff and packing dimension, respectively~\cite{Lutz03b,AHLM07}. These definitions were explicitly extended to Euclidean spaces in~\cite{LutMay08}, where the identities~(\ref{eq:dimDim}) were \emph{proven as a theorem}. Here it is convenient to use these identities as definitions. For $x \in \R^n$, it is easy to see that
		\[0\leq\dim(x)\leq\Dim(x)\leq n\,,\]
		and it is known that, for any two reals $0\leq\alpha\leq\beta\leq n$, there exist uncountably many 	points $x \in \R^n$ satisfying $\dim(x)=\alpha$ and $\Dim(x)=\beta$~\cite{AHLM07}. Applications of these dimensions in Euclidean spaces appear in~\cite{LutMay08,GuLuMa06,LutWei08,DLMT14,GLMM14}.
		\section{From Points to Sets}\label{sec:ps}
		The central message of this paper is a useful \emph{point-to-set principle} by which the existence of a single high-dimensional point in a set $E\subseteq\R^n$ implies that the set $E$ has high dimension.
		
		To formulate this principle we use relativization. All the algorithmic information concepts in Sections~\ref{sec:dpes} and~\ref{sec:cdes} above can be relativized to an arbitrary oracle $A\subseteq\N$ by giving the Turing machine in their definitions oracle access to $A$. Relativized Kolmogorov complexity $K_r^A(x)$ and relativized dimensions $\dim^A(x)$ and $\Dim^A(x)$ are thus well defined. Moreover, the results of Section~\ref{sec:dpes} hold relative to any oracle $A$.
	
		We first establish the point-to-set principle for Hausdorff dimension. Let $E\subseteq\R^n$. For $\delta>0$, define $\mathcal{U}_\delta(E)$ to be the collection of all countable covers of $E$ by sets of positive diameter at most $\delta$. That is, for every cover $\{U_i\}_{i\in\N}\in\mathcal{U}_\delta(E)$, we have $E\subseteq\bigcup_{i\in\N}U_i$ and $|U_i|\in(0,\delta]$ for all $i\in\N$, where for $X\in\R^n$, $|X|=\sup_{p,q\in X}|p-q|$. 		For $s\geq0$, define
		\[H_\delta^s(E)=\inf\bigg\{\sum_{i\in\N}\left|U_i\right|^s\,:\,\{U_i\}_{i\in\N}\in\mathcal{U}_\delta(E)\bigg\}\,.\]
		Then the \emph{$s$-dimensional Hausdorff outer measure of $E$} is
		\[H^s(E)=\lim_{\delta\to 0^+}H_\delta^s(E)\,,\]
		and the \emph{Hausdorff dimension of $E$} is
		\[\dim_H(E)=\inf\left\{s>0:H^s(E)=0\right\}\,.\] More details may be found in standard texts, e.g.,~\cite{SteSha05,Falc14}.
		
		\begin{thm}\label{thm:hausdorff}
			\textup{(Point-to-set principle for Hausdorff dimension)} For every set $E\subseteq\mathbb{R}^n$,
			\[\dim_H(E)=\adjustlimits\min_{A\subseteq\N}\sup_{x\in E}\,\dim^A(x)\,.\]
		\end{thm}
		
		Three things should be noted about this principle. First, while the left-hand side is the \emph{classical} Hausdorff dimension, which is a global property of $E$ that does not involve the theory of computing, the right-hand side is a pointwise property of the set that makes essential use of relativized algorithmic information theory. Second, as the proof shows, the right-hand side is a minimum, not merely an infimum. Third, and most crucially, this principle implies that, in order to prove a lower bound $\dim_H(E)\geq\alpha$, it suffices to show that, for every $A\subseteq\N$ and every $\ve>0$, there is a point $x\in E$ such that $\dim^A(x)\geq\alpha-\ve$.\footnote{The $\ve$ here is useful in general but is not needed in some cases, including our proof of Theorem~\ref{thm:Davies} below.}
		
		For the $(\geq)$ direction of this principle, we construct the minimizing oracle $A$. The oracle encodes, for a carefully chosen sequence of increasingly refined covers for $E$, the approximate locations and diameters of all cover elements. Using this oracle, a point $x\in\R^n$ can be approximated by specifying an appropriately small cover element that it belongs to, which requires an amount of information that depends on the number of similarly-sized cover elements. We use the definition of Hausdorff dimension to bound that number. The $(\leq)$ direction can be shown using results from~\cite{LutMay08}, but in the interest of self-containment we prove it directly.
		\begin{proofof}{Theorem~\ref{thm:hausdorff}}
			Let $E\subseteq\R^n$, and let $d=\dim_H(E)$. For every $s>d$ we have $H^s(E)=0$, so there is a sequence $\{\{U^{t,s}_i\}_{i\in\N}\}_{t\in\N}$ of countable covers of $E$ such that  $\big|U_i^{t,s}\big|\leq2^{-t}$ for every $i,t\in\N$, and for every sufficiently large $t$ we have
			\begin{equation}\label{eq:diam}
			\sum_{i\in\N} \left|U^{t,s}_i\right|^s<1\,.
			\end{equation}
			
			Let $D=\N^3\times(\Q\cap(d,\infty))$.
			Our oracle $A$ encodes functions
			$f_A:D\to\Q^n$
			and
			$g_A:D\to\Q$
			such that for every $(i,t,r,s)\in D$, we have \[f_A(i,t,r,s)\in B_{2^{-r-1}}(u)\]
			for some $u\in U_i^{t,s}$ and
			\begin{equation}\label{eq:diamerr}
				\Big|g_A(i,t,r,s)-\big|U_i^{t,s}\big|\Big|<2^{-r-4}\,.
			\end{equation}
			We will show, for every $x\in E$ and rational $s>d$, that $\dim^A(x)\leq s$.
			
			Fix $x\in E$ and $s\in\Q\cap(d,\infty)$. If for any $i_0,t_0\in\N$ we have $x\in U_{i_0}^{t_0,s}$ and $\left|U_{i_0}^{t_0,s}\right|=0$, then $U_{i_0}^{t_0,s}=\{x\}$, so $f_A(i_0,t_0,r,s)\in B_{2^{-r}}(x)$ for every $r\in\N$. In this case, let $M$ be a prefix Turing machine with oracle access to $A$ such that, whenever $U(\iota)=i\in\N$,  $U(\tau)=t\in\N$, $U(\rho)=r\in\N$, and $U(\sigma)=q\in\Q\cap(d,\infty)$,
			\[M(\iota\tau\rho\sigma)=f_A(i,t,r,q)\,.\]
			Now for any $r\in\N$, let $\iota$, $\tau$, $\rho$, and $\sigma$ be witnesses to $K(i_0)$, $K(t_0)$, $K(r)$, and $K(s)$, respectively.
			Since $i_0$, $t_0$, and $s$ are all constant in $r$ and $|\rho|=o(r)$, we have $|\iota\tau\rho\sigma|=o(r)$. Thus $K^A_r(x)=o(r)$, and $\dim^A(x)=0$. Hence assume that every cover element containing $x$ has positive diameter.
			
			Fix sufficiently large $t$, and let $U_{i_x}^{t,s}$ be some cover element containing $x$. 
			Let $M^\prime$ be a self-delimiting Turing machine with oracle access to $A$ such that whenever $U(\kappa)=k\in\N$, $U(\tau)=\ell\in\N$, $U(\rho)=r\in\N$, and $U(\sigma)=q\in\Q\cap(d,\infty)$,
			\[M^\prime(\kappa\tau\rho\sigma)=f_A(p,\ell,r,q)\,,\]
			where $p$ is the $k$\textsuperscript{th} index $i$ such that
			$g_A(i,t,r,q)\geq 2^{-r-3}$.
			
			Now fix $r\geq t-1$ such that
			\[\big|U_{i_x}^{t,s}\big|\in \left[2^{-r-2},2^{-r-1}\right)\,.\]
			Notice that $g_A(i_x,t,r,s)\geq 2^{-r-3}$. Hence there is some $k$ such that, letting $\kappa$, $\tau$, $\rho$, and $\sigma$ be witnesses to $K(k)$, $K(t)$, $K(r)$, and $K(s)$, respectively,
			\[M^\prime(\kappa\tau\rho\sigma)\in B_{2^{-r-1}}(u)\,,\]
			for some $u\in U^{t,s}_{i_x}$.
			Because $\big|U_{i_x}^{t,s}\big|<2^{-r-1}$ and $x\in	U_{i_x}^{t,s}$, we have \[M^\prime(\kappa\tau\rho\sigma)\in B_{2^{-r}}(x)\,.\]
			Thus
			\[K^A_{r}(x)\leq K(k)+K(t)+K(s)+K(r)+c\,,\]
			where $c$ is a machine constant for $M^\prime$. Since $s$ is constant in $r$ and $t<r$, Observation~\ref{obs:KBoundN} tells us that this expression is $K(k)+o(r)\leq\log(k)+o(r)$. By~(\ref{eq:diam}), there are fewer than $2^{(r+4)s}$ indices $i\in\N$ such that
			\[\left|U_i^{t,s}\right|\geq2^{-r-4}\,,\]
			hence by~(\ref{eq:diamerr}) there are fewer than $2^{(r+4)s}$ indices $i\in\N$ such that $g_A(i,t,r,s)\geq 2^{-r-3}$,
			so $\log(k)<{(r+4)s}$.
			Therefore $K_{r}^A(x)\leq rs+o(r)$.
			
			There are infinitely many such $r$, which can be seen by replacing $t$ above with $r+2$. We have shown
			\[\dim^A(x)=\liminf_{r\to\infty}\frac{K_r^A(x)}{r}\leq s\,,\]
			for every rational $s>d$, hence $\dim^A(x)\leq d$. It follows that
			\[\adjustlimits\min_{A\subseteq\N}\sup_{x\in E}\,\dim^A(x)\leq d\,.\]
			
			For the other direction, assume for contradiction that there is some oracle $A$ and $d^\prime<d$ such that
			\[\sup_{x\in E}\,\dim^A(x)= d^\prime\,.\]
			Then for every $x\in E$, $\dim^A(x)\leq d^\prime$.  Let $s\in(d^\prime,d)$. For every $r\in\N$, define the sets
			\[\mathcal{B}_r=\left\{B_{2^{-r}}(q)\,:\,q\in\Q\textrm{ and }K^A(q)\leq rs\right\}\]
			and
			\[\mathcal{W}_r=\bigcup_{k=r}^\infty\mathcal{B}_k\,.\]
			There are at most $2^{ks+1}$ balls in each $\mathcal{B}_k$, so for every $r\in\N$ and $s^\prime\in(s,d)$,
			\begin{align*}
			\sum_{W\in \mathcal{W}_r}|W|^{s^\prime}&=\sum_{k=r}^\infty\sum_{W\in\mathcal{B}_k}|W|^{s^\prime}\\
			&\leq\sum_{k=r}^\infty 2^{ks+1}(2^{1-k})^{s^\prime}\\
			&=2^{1+s'}\cdot\sum_{k=r}^\infty 2^{(s-s^\prime)k}\,,
			\end{align*}
			which approaches $0$ as $r\to\infty$. As every $\mathcal{W}_r$ is a cover for $E$, we have $H^{s^\prime}(E)=0$, so $\dim_H(E)\leq s^\prime<d$, a contradiction.\qed
		\end{proofof}
		
		The \emph{packing dimension} $\dim_P(E)$ of a set $E\subseteq\R^n$, defined in the appendix and standard texts, e.g.,~\cite{Falc14}, is a dual of Hausdorff dimension satisfying $\dim_P(E)\geq\dim_H(E)$, with equality for very ``regular'' sets $E$. We also have the following.
		\begin{thm}\label{thm:packing}
			\textup{(Point-to-set principle for packing dimension)} For every set $E\subseteq\R^n$,
			\[\dim_P(E)=\adjustlimits\min_{A\subseteq\N}\sup_{x\in E}\,\Dim^A(x)\,.\]
		\end{thm}

		\section{Conditional Kolmogorov Complexity in Euclidean Spaces}\label{sec:ckces}
		We now develop the conditional Kolmogorov complexity in Euclidean spaces.
		
		For $x \in \R^m$, $q\in\Q^n$, and $r\in\N$, the \emph{conditional Kolmogorov complexity} of $x$ at \emph{precision} $r$ \emph{given} $q$ is
		\begin{equation}\label{eq:Krxq}
			\hat{K}_r(x|q)=\min\left\{K(p|q)\,:\,p\in\Q^m\cap B_{2^{-r}}(x)\right\}\,.
		\end{equation}
		For $x\in\R^m$, $y\in\R^n$, and $r,s\in\N$, the \emph{conditional Kolmogorov complexity} of $x$ at \emph{precision} $r$ \emph{given} $y$ at \emph{precision} $s$ is
		\begin{equation}\label{eq:Krsxy}
		K_{r,s}(x|y)=\max\big\{\hat{K}_r(x|q)\,:\,q\in\Q^n\cap B_{2^{-s}}(y)\big\}\,.
		\end{equation}
		Intuitively, the maximizing argument $q$ is the point near $y$ that is least helpful in the task of approximating $x$.
		Note that $K_{r,s}(x|y)$ is finite, because $\hat{K}_r(x|q)\leq K_r(x)+O(1)$. For $x\in\R^m$, $y\in\R^n$, and $r\in\N$, the \emph{conditional Kolmogorov complexity} of $x$ \emph{given} $y$ at \emph{precision} $r$ is
		\begin{equation}\label{eq:Krxy}
		K_r(x|y)=K_{r,r}(x|y)\,.
		\end{equation}
		
		\begin{thm}\label{thm:KrChainRule}\textup{({Chain rule for $K_r$})}
			For all $x\in\R^m$ and $y\in\R^n$,
			\[K_r(x,y)=K_r(x|y) + K_r(y)+o(r)\,.\]
		\end{thm}

		We also consider the Kolmogorov complexity of $x\in\R^m$ at precision $r$ \emph{relative} to $y\in\R^n$. Let $K^y_r(x)$ denote $K^{A_y}_r(x)$, where $A_y\subseteq\N$ encodes the binary expansions of $y$'s coordinates. The following lemma reflects the intuition that oracle access to $y$ is at least as useful as any bounded-precision estimate for $y$.
		\begin{lem}\label{lem:relcond}
			For each $m,n\in\N$ there is a constant $c\in\N$ such that, for all $x\in\R^m$, $y\in\R^n$, and $r,s\in\N$,
			\[K_r^y(x)\leq K_{r,s}(x|y)+K(s)+c\,.\]
			In particular, $K_r^y(x)\leq K_r(x|y)+K(r)+c$.
		\end{lem}
		\section{Kakeya Sets in the Plane}\label{sec:ksp}
			This section uses the results of the preceding two sections to give a new proof of the following classical theorem. Recall that a \emph{Kakeya set} in $\R^n$ is a set containing a unit line segment in every direction.

		\begin{thm}\textup{(Davies \cite{Davi71})}\label{thm:Davies}
			Every Kakeya set in $\mathbb{R}^2$ has Hausdorff dimension 2.
		\end{thm}

		Our new proof of Theorem~\ref{thm:Davies} uses a relativized version of the following lemma.
		\begin{lem}\label{lem:main}
			Let $m\in[0,1]$ and $b\in\mathbb{R}$. Then for almost every $x\in[0,1]$, \begin{equation}\label{eq:mainlem}
			\liminf_{r\to\infty}\frac{K_r(m,b,x)-K_r(b|m)}{r}\leq\dim(x,mx+b)\,.
			\end{equation}
		\end{lem}
		
		\begin{proof}
			We build a program that takes as input a precision level $r$, an approximation $p$ of $x$, an approximation $q$ of $mx+b$, a program $\pi$ that will approximate $b$ given an approximation for $m$, and a natural number $h$.
			In parallel, the program considers each multiple of $2^{-r}$ in [0,1] as a possible approximate value $u$ for the slope $m$, and it checks whether each such $u$ is consistent with the program's inputs. If $u$ is close to $m$, then $\pi(u)$ will be close to $b$, so $up+\pi(u)$ will be close to $mx+b$. Any $u$ that satisfies this condition is considered a ``candidate'' for approximating $m$.
			
			Some of these candidates may be ``false positives,'' in that there can be values of $u$ that are far from $m$ but for which $up+\pi(u)$ is still close to $mx+b$. Thus the program is also given an input $h$ so that it can choose the correct candidate; it selects the $h$\textsuperscript{th} candidate that arises in its execution. We will show that this $h$ is often not large enough to significantly affect the total input length.
			
			Formally, let $M$ be a Turing machine that runs the following algorithm on input $\rho\pi\sigma\eta$ whenever $U(\rho)=r\in\N$, $U(\eta)=h\in\N$, and $U(\sigma)=(p,q)\in\mathbb{Q}^2$:
			
			\begin{itemize}
				\item[] $candidate:=0$
				\item[] for $i=0,1,\ldots,2^r$, in parallel:
				\begin{itemize}
					\item[] $u_i:=2^{-r}i$
					\item[] $v_i:=U(\pi,u_i)$
					\item[] do atomically:
					\begin{itemize}
						\item[] if $v_i\in\R$ and $|u_ip+v_i-q|<2^{2-r}$, then $candidate:=candidate+1$
						\item[] if $candidate = h$, then return $(u_i,v_i,p)$ and halt
					\end{itemize}
				\end{itemize}
			\end{itemize}
			
			Fix $m\in[0,1]$ and $b\in\mathbb{R}$. For each $r\in\mathbb{N}$, let $m_r=2^{-r}\lfloor m\cdot2^{r}\rfloor$, and fix $\pi_r$ testifying to the value of $\hat{K}_r(b|m_r)$ and $\sigma_r$ testifying to the value of $K_r(x,mx+b)$. 
			
			Proofs of the following four claims appear in the appendix. Intuitively, Claim~\ref{cla:floor} says that no point in $B_{2^{-r}}(m)$ gives much less information about $b$ than $m_r$ does. Claim~\ref{cla:ExistsH} states that there is always some value of $h$ that causes this machine to return the desired output. Claim~\ref{cla:Meas0} says that for almost every $x$, this value does not grow too quickly with $r$, and Claim~\ref{cla:pointdim} says that (\ref{eq:mainlem}) holds for every such $x$.
			\begin{cla}\label{cla:floor}
				For every $r\in\N$, $K_r(b|m)=\hat{K}_r(b|m_r)+o(r)$.
			\end{cla}
			\begin{cla}\label{cla:ExistsH}
				For each $x\in[0,1]$ and $r\in\mathbb{N}$, there exists an $h\in\mathbb{N}$ such that 
				\[M(\rho\pi_r\sigma_r\eta)\in B_{2^{1-r}}(m,b,x)\,,\]
				where $U(\rho)=r$ and $U(\eta)=h$.
			\end{cla}
			For every $x\in[0,1]$ and $r\in\N$, define $h(x,r)$ to be the minimal $h$ satisfying the conditions of Claim~\ref{cla:ExistsH}.
			\begin{cla}\label{cla:Meas0}
				For almost every $x\in[0,1]$, $\log(h(x,r))=o(r)$.
			\end{cla}
			\begin{cla}\label{cla:pointdim}
				For every $x\in[0,1]$, if $\log(h(x,r))={o(r)}$, then
				\[\liminf_{r\to\infty}\frac{K_r(m,b,x)-K_r(b|m)}{r}\leq\dim(x,mx+b)\,.\]
			\end{cla}
			
			The lemma follows immediately from Claims~\ref{cla:Meas0} and~\ref{cla:pointdim}.
		\end{proof}

		\begin{proofof}{Theorem~\ref{thm:Davies}}
			Let $K$ be a Kakeya set in $\R^2$. By Theorem~\ref{thm:hausdorff}, there exists an oracle $A$ such that $\dim_H(K)=\sup_{p\in K}\dim^A(p)$.
			
			Let $m\in[0,1]$ such that $\dim^A(m)=1$; such an $m$ exists by Theorem 4.5 of~\cite{Lutz03b}. $K$ contains a unit line segment $L$ of slope $m$. Let $(x_0,y_0)$ be the left endpoint of such a segment. Let $q\in\Q\cap[x_0,x_0+1/8]$, and let $L^\prime$ be the unit segment of slope $m$ whose left endpoint is $(x_0-q,y_0)$. Let $b=y_1+qm$, the $y$-intercept of $L^\prime$. 
			
			By a relativized version of Lemma~\ref{lem:main}, there is some $x\in[0,1/2]$ such that $\dim^{A,m,b}(x)=1$ and
			\[\liminf_{r\to\infty}\frac{K^A_r(m,b,x)-K^A_r(b|m)}{r}\leq\dim^A(x,mx+b)\,.\]
			(This holds because almost every $x\in[0,1/2]$ is algorithmically random relative to $(A,m,b)$ and hence satisfies $\dim^{A,m,b}(x)=1$.)
			Fix such an $x$, and notice that $(x,mx+b)\in L^\prime$. Now applying a relativized version of Theorem~\ref{thm:KrChainRule},
			\begin{align*}
			\dim^A(x,mx+b)&\geq\liminf_{r\to\infty}\frac{K^A_r(m,b,x)-K^A_r(b|m)}{r}\\
			&=\liminf_{r\to\infty}\frac{K^A_r(m,b,x)-K^A_r(b,m)+K^A_r(m)}{r}\\
			&=\liminf_{r\to\infty}\frac{K^A_r(x|b,m)+K^A_r(m)}{r}\\
			&\geq\liminf_{r\to\infty}\frac{K^A_r(x|b,m)}{r}+\liminf_{r\to\infty}\frac{K^A_r(m)}{r}\,.
			\end{align*}
			By Lemma~\ref{lem:relcond}, $K^A_r(x|b,m)\geq K^{A,b,m}_r(x)+o(r)$, so we have
			\begin{align*}
			\dim^A(x,mx+b)&\geq\liminf_{r\to\infty}\frac{K^{A,b,m}_r(x)}{r}+\liminf_{r\to\infty}\frac{K^A_r(m)}{r}\\
			&=\dim^{A,b,m}(x)+\dim^A(m)\,,
			\end{align*}
			which is $2$ by our choices of $m$ and $x$.
			
			By Observation~\ref{obs:addrat},
			\[\dim^A(x,mx+b)=\dim^A(x+q,mx+b)\,.\] Hence, there exists a point $(x+q,mx+b)\in K$ such that $\dim^A(x+q,mx+b)\geq 2$. By Theorem~\ref{thm:hausdorff}, the point-to-set principle for Hausdorff dimension, this completes the proof.
			\qed
		\end{proofof}
		
			It is natural to ask what prevents us from extending this proof to higher-dimensional Euclidean spaces. The point of failure in a direct extension would be Claim~\ref{cla:Meas0} in the proof of Lemma~\ref{lem:main}. Speaking informally, the problem is that the total number of candidates may grow as $2^{(n-1)r}$, meaning that $\log(h(x,r))$ could be $\Omega((n-2)r)$ for every $x$.
		\section{Conditional Dimensions in Euclidean Spaces}\label{sec:cdes}
			The results of Section~\ref{sec:ckces}, which were used in the proof of Theorem~\ref{thm:Davies}, also enable us to give robust formulations of conditional dimensions.
		
			For $x\in\R^m$ and $y\in\R^n$, the \emph{lower} and \emph{upper conditional dimensions} of $x$ \emph{given} $y$ are
			\begin{equation}\label{eq:conddimDim}
			\dim(x|y)=\liminf_{r\to\infty}\frac{K_r(x|y)}{r}\qquad\textrm{and}\qquad\Dim(x|y)=\limsup_{r\to\infty}\frac{K_r(x|y)}{r}\,,
			\end{equation}
			respectively.	
			
			The use of the same precision bound $r$ for both $x$ and $y$ in~(\ref{eq:Krxy}) makes the definitions~(\ref{eq:conddimDim}) appear arbitrary and ``brittle.''  The following theorem shows that this is not the case.
			
			\begin{thm}\label{thm:CondDimRobust}
				Let $s:\N\to\N$. If $|s(r)-r|=o(r)$, then, for all $x\in\R^m$ and $y\in\R^n$,
				\[\dim(x|y)=\liminf_{r\to\infty}\frac{K_{r,s(r)}(x|y)}{r}\,,\]
				and
				\[\Dim(x|y)=\limsup_{r\to\infty}\frac{K_{r,s(r)}(x|y)}{r}\,.\]
			\end{thm}
					
			The rest of this section is devoted to showing that our conditional dimensions have the correct information theoretic relationships with the previously developed dimensions and mutual dimensions.
			
			Mutual dimensions were developed very recently, and Kolmogorov complexity was the starting point. The \emph{mutual (algorithmic) information} between two strings $u,v \in \{0,1\}^*$ is
			\[I(u:v) = K(v) - K(v|u)\,.\]
			Again, routine coding extends $K(u|v)$ and $I(u:v)$ to other discrete domains.  Discussions of $K(u|v)$, $I(u:v)$, and the correspondence of $K(u)$, $K(u|v)$, and $I(u:v)$ with Shannon entropy, Shannon conditional entropy, and Shannon mutual information appear in~\cite{LiVit08}.
			
			In parallel with~(\ref{eq:Kr}) and~(\ref{eq:dimDim}), Case and J. H. Lutz~\cite{CasLut15} lifted the definition of $I(p:q)$ for rational points $p$ and $q$ in Euclidean spaces in two steps to define the mutual dimensions between two arbitrary points in (possibly distinct) Euclidean spaces.  First, for $x \in \R^m$, $y \in \R^n$, and $r \in \N$, the \emph{mutual information} between $x$ and $y$ at \emph{precision} $r$ is
			\begin{equation}\label{eq:Ir}
			I_r(x:y)=\min\left\{I(p:q)\,:\,p\in B_{2^{-r}}(x)\cap\Q^m\textrm{ and }q\in B_{2^{-r}}(y)\cap\Q^n\right\}\,,
			\end{equation}
			where $B_{2^{-r}}(x)$ and $B_{2^{-r}}(y)$ are the open balls of radius $2^{-r}$ about $x$ and $y$ in their respective Euclidean spaces.  Second, for $x \in \R^m$ and $y \in \R^n$, the \emph{lower} and \emph{upper mutual dimensions} between $x$ and $y$ are
			\begin{equation}\label{eq:mdimMdim}
			\mdim(x:y)=\liminf_{r\to\infty}\frac{I_r(x:y)}{r}\quad \textrm{and}\quad\Mdim(x:y)=\limsup_{r\to\infty}\frac{I_r(x:y)}{r}\,,
			\end{equation}
			respectively.  Useful properties of these mutual dimensions, especially including data processing inequalities, appear in~\cite{CasLut15}.
			
			\begin{lem}\label{lem:MutualInfo}
				For all $x\in\R^m$ and $y\in\R^n$,
				\[I_r(x:y)=K_r(x)-K_r(x|y)+o(r)\,.\]
			\end{lem}
			The following bounds on mutual dimension follow from Lemma~\ref{lem:MutualInfo}.
			\begin{thm}
				For all $x\in\R^m$ and $y\in\R^n$, the following hold.
				\begin{enumerate}
					\item $\mdim(x:y)\geq \dim(x)-\Dim(x|y)$.
					\item $\Mdim(x:y)\leq \Dim(x)-\dim(x|y)$.
				\end{enumerate}
			\end{thm}
			Our final theorem is easily derived from Theorem~\ref{thm:KrChainRule}.
			\begin{thm}\label{thm:DimChainRule}
				\textup{(Chain rule for dimension)}
				For all $x\in\R^m$ and $y\in\R^n$,
				\begin{align*}
				\dim(x)+\dim(y|x)&\leq\dim(x,y)\\
				&\leq\dim(x)+\Dim(y|x)\\
				&\leq\Dim(x,y)\\
				&\leq\Dim(x)+\Dim(y|x)\,.
				\end{align*}
			\end{thm}
	\section{Conclusion}\label{sec:conc}
		This paper shows a new way in which theoretical computer science can be used to answer questions that may appear unrelated to computation. We are hopeful that our new proof of Davies's theorem will open the way for using constructive fractal dimensions to make new progress in geometric measure theory, and that conditional dimensions will be a useful component of the information theoretic apparatus for studying dimension. 
		\subsection*{Acknowledgments}
			We thank Eric Allender for useful corrections and three anonymous reviewers of an earlier version of this work for helpful input regarding presentation.
		\bibliography{aipkscd}

\begin{thebibliography}{10}

\bibitem{AHLM07}
Krishna~B. Athreya, John~M. Hitchcock, Jack~H. Lutz, and Elvira Mayordomo.
\newblock {Effective strong dimension in algorithmic information and
  computational complexity}.
\newblock {\em SIAM J. Comput.}, 37(3):671--705, 2007.

\bibitem{Besi19}
A.~S. Besicovitch.
\newblock Sur deux questions d'int\'egrabilit\'e des fonctions.
\newblock {\em Journal de la Soci\'{e}t\'{e} de physique et de mathematique de
  l'Universite de Perm}, 2:105--123, 1919.

\bibitem{Besi28b}
A.~S. Besicovitch.
\newblock On {K}akeya's problem and a similar one.
\newblock {\em Mathematische Zeitschrift}, 27:312--320, 1928.

\bibitem{CasLut15}
Adam Case and Jack~H. Lutz.
\newblock Mutual dimension.
\newblock {\em ACM Transactions on Computation Theory}, 7(3):12, 2015.

\bibitem{CasLut15a}
Adam Case and Jack~H. Lutz.
\newblock Mutual dimension and random sequences.
\newblock In Giuseppe~F. Italiano, Giovanni Pighizzini, and Donald Sannella,
  editors, {\em Mathematical Foundations of Computer Science 2015 - 40th
  International Symposium, {MFCS} 2015, Milan, Italy, August 24-28, 2015,
  Proceedings, Part {II}}, volume 9235 of {\em Lecture Notes in Computer
  Science}, pages 199--210. Springer, 2015.

\bibitem{Chai66}
Gregory~J. Chaitin.
\newblock On the length of programs for computing finite binary sequences.
\newblock {\em J. {ACM}}, 13(4):547--569, 1966.

\bibitem{Chai69}
Gregory~J. Chaitin.
\newblock On the length of programs for computing finite binary sequences:
  statistical considerations.
\newblock {\em J. {ACM}}, 16(1):145--159, 1969.

\bibitem{CovTho06}
Thomas~R. Cover and Joy~A. Thomas.
\newblock {\em Elements of Information Theory}.
\newblock Wiley, second edition, 2006.

\bibitem{Davi71}
Roy~O. Davies.
\newblock Some remarks on the {K}akeya problem.
\newblock {\em Proc. Cambridge Phil. Soc.}, 69:417--421, 1971.

\bibitem{DLMT14}
Randall Dougherty, Jack~H. Lutz, R.~Daniel Mauldin, and Jason Teutsch.
\newblock Translating the {C}antor set by a random real.
\newblock {\em Transactions of the American Mathematical Society},
  366:3027--3041, 2014.

\bibitem{DowHir10}
Rod Downey and Denis Hirschfeldt.
\newblock {\em Algorithmic Randomness and Complexity}.
\newblock Springer-Verlag, 2010.

\bibitem{Dvir09}
Zeev Dvir.
\newblock On the size of {K}akeya sets in finite fields.
\newblock {\em J. Amer. Math. Soc.}, 22:1093--1097, 2009.

\bibitem{Falc14}
Kenneth Falconer.
\newblock {\em Fractal Geometry: Mathematical Foundations and Applications}.
\newblock Wiley, third edition, 2014.

\bibitem{GuLuMa06}
Xiaoyang Gu, Jack~H. Lutz, and Elvira Mayordomo.
\newblock Points on computable curves.
\newblock In {\em FOCS}, pages 469--474. IEEE Computer Society, 2006.

\bibitem{GLMM14}
Xiaoyang Gu, Jack~H. Lutz, Elvira Mayordomo, and Philippe Moser.
\newblock Dimension spectra of random subfractals of self-similar fractals.
\newblock {\em Ann. Pure Appl. Logic}, 165(11):1707--1726, 2014.

\bibitem{Haus19}
Felix Hausdorff.
\newblock Dimension und \"ausseres {M}ass.
\newblock {\em Mathematische Annalen}, 79:157--179, 1919.

\bibitem{Kolm65}
Andrei~N. Kolmogorov.
\newblock Three approaches to the quantitative definition of information.
\newblock {\em Problems of Information Transmission}, 1(1):1--7, 1965.

\bibitem{Levi73}
Leonid~A. Levin.
\newblock On the notion of a random sequence.
\newblock {\em Soviet Math Dokl.}, 14(5):1413--1416, 1973.

\bibitem{Levi74}
Leonid~A. Levin.
\newblock Laws of information conservation (nongrowth) and aspects of the
  foundation of probability theory.
\newblock {\em Problemy Peredachi Informatsii}, 10(3):30--35, 1974.

\bibitem{LiVit08}
Ming Li and Paul~M.B. Vit\'{a}nyi.
\newblock {\em An Introduction to {K}olmogorov Complexity and Its
  Applications}.
\newblock Springer, third edition, 2008.

\bibitem{Lutz03a}
Jack~H. Lutz.
\newblock Dimension in complexity classes.
\newblock {\em {SIAM} J. Comput.}, 32(5):1236--1259, 2003.

\bibitem{Lutz03b}
Jack~H. Lutz.
\newblock The dimensions of individual strings and sequences.
\newblock {\em Inf. Comput.}, 187(1):49--79, 2003.

\bibitem{LutLut15a}
Jack~H. Lutz and Neil Lutz.
\newblock Lines missing every random point.
\newblock {\em Computability}, 4(2):85--102, 2015.

\bibitem{LutMay08}
Jack~H. Lutz and Elvira Mayordomo.
\newblock Dimensions of points in self-similar fractals.
\newblock {\em SIAM J. Comput.}, 38(3):1080--1112, 2008.

\bibitem{LutWei08}
Jack~H. Lutz and Klaus Weihrauch.
\newblock Connectivity properties of dimension level sets.
\newblock {\em Mathematical Logic Quarterly}, 54:483--491, 2008.

\bibitem{Mars54}
John~M. Marstrand.
\newblock Some fundamental geometrical properties of plane sets of fractional
  dimensions.
\newblock {\em Proceedings of the London Mathematical Society}, 4(3):257--302,
  1954.

\bibitem{Mayo02}
Elvira Mayordomo.
\newblock A {K}olmogorov complexity characterization of constructive
  {H}ausdorff dimension.
\newblock {\em Inf. Process. Lett.}, 84(1):1--3, 2002.

\bibitem{Nies09}
Andre Nies.
\newblock {\em Computability and Randomness}.
\newblock Oxford University Press, Inc., New York, NY, USA, 2009.

\bibitem{Shan48}
Claude~E. Shannon.
\newblock A mathematical theory of communication.
\newblock {\em Bell System Technical Journal}, 27(3--4):379--423, 623--656,
  1948.

\bibitem{Solo64}
Ray~J. Solomonoff.
\newblock A formal theory of inductive inference.
\newblock {\em Information and Control}, 7(1-2):1--22, 224--254, 1964.

\bibitem{SteSha05}
Elias~M. Stein and Rami Shakarchi.
\newblock {\em Real Analysis: Measure Theory, Integration, and Hilbert Spaces}.
\newblock Princeton Lectures in Analysis. Princeton University Press, 2005.

\bibitem{Tao00}
Terence Tao.
\newblock From rotating needles to stability of waves: emerging connections
  between combinatorics, analysis, and {PDE}.
\newblock {\em Notices Amer. Math. Soc}, 48:294--303, 2000.

\bibitem{Turi37}
Alan~M. Turing.
\newblock On computable numbers, with an application to the
  {E}ntscheidungsproblem. {A} correction.
\newblock {\em Proceedings of the London Mathematical Society}, 43(2):544--546,
  1937.

\bibitem{Weih00}
Klaus Weihrauch.
\newblock {\em Computable Analysis: An Introduction}.
\newblock Springer, 2000.

\bibitem{Wolf99}
T.~Wolff.
\newblock Recent work connected with the {K}akeya problem.
\newblock {\em Prospects in Mathematics}, pages 129--162, 1999.

\end{thebibliography}
		\clearpage
		\pagenumbering{arabic}
		\renewcommand*{\thepage}{A\arabic{page}}
		\appendix
		\section{Appendix}\label{app:main}
			\renewcommand\thethm{\thesection.\arabic{thm}}
			\setcounter{thm}{0}
			\subsection{Packing dimension}
				Let $E\subseteq\R^n$. For $\delta>0$, define $\mathcal{V}_\delta(E)$ to be the collection of all countable \emph{packings} of $E$ by disjoint open balls of diameter at most $\delta$. That is, for every packing $\{V_i\}_{i\in\N}\in\mathcal{V}_\delta(E)$ and every $i\in\N$, we have $V_i=B_{\ve_i}(x_i)\subseteq E$ for some $x_i\in E$ and $\ve_i\in[0,\delta/2]$.
				
				For $s\geq0$, define
				\[P_\delta^s(E)=\sup\bigg\{\sum_{i\in\N}\left|V_i\right|^s\,:\,\{V_i\}_{i\in\N}\in\mathcal{V}_\delta(E)\bigg\}\,,\]
				and let
				\[P_0^s(E)=\lim_{\delta\to0^+}P^s_\delta(E)\,.\]
				Then the \emph{$s$-dimensional packing measure of $E$} is
				\[P^s(E)=\inf\bigg\{\sum_{i\in\N} P_0^s(E_i)\,:\,E\subseteq\bigcup_{i\in\N}E_i\bigg\}\,,\]
				and the \emph{packing dimension of $E$} is
				\[\dim_P(E)=\inf\left\{s:P^s(E)=0\right\}\,.\]
				{
					\renewcommand{\thethm}{\ref{thm:packing}}
					\begin{thm}
						\textup{(Point-to-set principle for packing dimension)} For every set $E\subseteq\mathbb{R}^n$,
						\[\dim_P(E)=\adjustlimits\min_{A\subseteq\N}\sup_{x\in E}\,\Dim^A(x)\,.\]
					\end{thm}
					\addtocounter{thm}{-1}
				}
				\begin{proof}
					Let $E\subseteq\R^n$, and let $d=\dim_P(E)$. For every $s>d$ we have $P^s(E)=0$, so there is a cover $\big\{E^{s}_j\big\}_{j\in\N}$ for $E$ such that
					\begin{equation}\label{eq:thm:packing:1}
					\sum_{j\in\N}\lim_{\delta\to 0^+}P_\delta^s(E^{s}_j)<1\,.
					\end{equation}
					For every $r,j\in\N$, let
					\[\big\{V_i^{r,s,j}\big\}_{i\in\N}\in\mathcal{V}_{2^{-r-2}}(E_j^s)\]
					be a maximal packing of $E_j^s$ by open balls of radius exactly $2^{-r-2}$ (and higher-indexed balls of radius $0$).

					Let $D=\N^3\times(\Q\cap(d,\infty))$. Our oracle $A$ encodes a function $f_A:D\to\Q^n$
					such that for every $(i,j,r,s)\in D$ we have
					\[f_A(i,j,r,s)\in V_i^{r,s,j}\,.\]
					We will show, for every $x\in E$ and rational $s>d$, that $\Dim^A(x)\leq s$.
					
					Let $M$ be a self-delimiting Turing machine with oracle access to $A$ such that, whenever $U(\iota)=i\in\N$, $U(\kappa)=j\in\N$, $U(\rho)=r\in\N$, and $U(\sigma)=q\in\Q\cap(d,\infty)$,
					\[M(\iota\kappa\rho\sigma)=f_A(i,j,r,s)\,.\]
					
					Fix $x\in E$ and $s\in\Q\cap(0,\infty)$, and let $k\in\N$ be such that $x\in E_{k}^{s}$.
					Notice that by our choice of packing, for every $r\in\N$ there must be some $i_r\in\N$  such that
					\[V_{i_r}^{r,s,k}\subseteq B_{2^{-r}}(x)\,.\]
					Thus, for every $r\in\N$, letting $\iota$, $\kappa$, $\rho$, $\sigma$ testify to $K(i_r)$, $K(k)$, $K(r)$, and $K(s)$, respectively,
					\begin{align*}
					M(\iota\kappa\rho\sigma)&=f_A(i_r,k,r,s)\\
					&\in V_{i_r}^{r,s,k}\\
					&\subseteq B_{2^{-r}}(x)\,,
					\end{align*}
					hence $K_r^A(x)\leq K(i_r)+K(k)+K(r)+K(s)+c$, where $c$ is a machine constant for $M$.
					Because $k$ and $s$ are constant in $r$, $K(r)=o(r)$, and $K(i_r)\leq\log i_r+o(r)$, we have
					\[K_{r}^A(x)\leq \log i_r + o(r)\,.\]
								
					By~(\ref{eq:thm:packing:1}), $\lim_{\delta\to 0^+}P^s_\delta(E_k^s)<1$, so there is some $R\in\N$ such that, for every $r>R$, $P^s_{2^{-r}}(E_k^s)<1$. Then for every $r>R$,
					\[\sum_{i\in\N}\big|V_i^{r,s,k}\big|^s<1\,,\]
					hence there are fewer than $2^{(r+2)s}$ balls of radius $2^{-r-2}$ in the packing, and $\log i_r<(r+2)s$. We conclude that $K_r^A(x)\leq rs+o(r)$ for every $r>R$, so
					\[\dim^A(x)=\limsup_{r\to\infty}\frac{K_r^A(x)}{r}\leq s\,.\]
					Since this holds for every rational $s>d$, we have shown $\Dim^A(x)\leq d$ and thus \[\adjustlimits\min_{A\subseteq\N}\sup_{x\in E}\,\Dim^A(x)\leq d\,.\]

					For the other direction, assume for contradiction that there is some oracle $A$ and $d^\prime<d$ such that
					\[\sup_{x\in E}\,\Dim^A(x)= d^\prime\,.\]
					Then for every $x\in E$, $\Dim^A(x)\leq d^\prime$.  Let $s\in(d^\prime,d)$. For every $k\in\N$, define the set
					\[C_k=\bigcup\left\{B_{2^{-k}}(q)\,:\,q\in\Q\textrm{ and }K^A(q)\leq ks\right\}\,,\]
					and for every $i\in\N$, define
					\[E_i=\bigcap_{k=i}^\infty C_k\,.\]
					
					For $r\geq i$, consider any packing in
					$\mathcal{V}_{2^{-r}}\left(E_i\right)$. Let $B_\ve(x)$ be an element of the packing, and let $k=\lceil-\log\ve\rceil$. Then $k\geq r+1>i$, so $B_\ve(x)\subseteq E_i\subseteq C_k$. In particular $x\in C_k$, meaning that there is some $q\in\Q$ such that $K^A(q)\leq ks$ and $x\in B_{2^{-k}}(q)$. As $2^{-k}\leq\ve$, we also have $q\in B_\ve(x)$. Thus, every packing element of radius at least $2^{-k}$ contains a (distinct) member of the set $\{q\in\Q:K^A(q)\leq ks\}$. It follows that for every $k\geq r+1$, the packing includes at most $2^{ks+1}$ elements with diameters in the range $[2^{1-k},2^{2-k})$.
					
					Now let $s'\in(s,d)$. For every $i\in\N$ and $r\geq i$, we have
					\begin{align*}
					P_{2^{-r}}^{s'}(E_i)&=\sup\bigg\{\sum_{j\in\N}\left|V_j\right|^{s'}\,:\,\{V_j\}_{j\in\N}\in\mathcal{V}_{2^{-r}}(E_i)\bigg\}\\
					&\leq\sum_{k=r+1}^\infty 2^{ks+1}(2^{2-k})^{s'}\\
					&=2^{1+2s'}\cdot\sum_{k=r+1}^\infty 2^{(s-s')k}\,.
					\end{align*}
					This approaches $0$ as $r\to\infty$, so $P_0^{s'}(E_i)=0$. Observe now that
					\[E\subseteq\bigcup_{i\in\N}E_i\,.\]
					Thus,
					\[P^{s'}(E)\leq\sum_{i\in\N}P^{s'}_0(E_i)=0\,,\]
					meaning that $\dim_P(E)\leq s^\prime<d$, a contradiction. We conclude that for every oracle $A$, \[\sup_{x\in E}\,\Dim^A(x)\geq d\,.\]
				\end{proof}
			\subsection{Chain rule for $K_r$}
				{
					\renewcommand{\thethm}{\ref{thm:KrChainRule}}
					\begin{thm}
						For all $x\in\R^m$, $y\in\R^n$, and $r\in\N$,
						\[K_r(x,y)=K_r(x|y) + K_r(y)+o(r)\,.\]
					\end{thm}
					\addtocounter{thm}{-1}
				}
				\begin{proof}
					Theorem 4.10 of~\cite{CasLut15} tells us that
					\[I_r(x:y)=K_r(x)+K_r(y)-K_r(x,y)+o(r)\,.\]
					Combining this with Lemma~\ref{lem:MutualInfo}, we have
					\[K_r(x)+K_r(y)-K_r(x,y)+o(r)=K_r(x)-K_r(x|y)+o(r)\,.\]
					The theorem follows immediately.
				\end{proof}

			\subsection{Proof of Lemma~\ref{lem:relcond}}
			{
				\renewcommand{\thethm}{\ref{lem:relcond}}
				\begin{lem}
					For each $m,n\in\N$ there is a constant $c\in\N$ such that, for all $x\in\R^m$, $y\in\R^n$, and $r,s\in\N$,
					\[K_r^y(x)\leq K_{r,s}(x|y)+K(s)+c\,.\]
					In particular, $K_r^y(x)\leq K_r(x|y)+K(r)+c$.
				\end{lem}
				\addtocounter{thm}{-1}
			}
			\begin{proof}
				Let $m,n\in\N$, and let $U$ be the optimal Turing machine fixed for the definition of conditional Kolmogorov complexity. Let $M$ be an oracle Truing machine that, on input $\pi\in\{0,1\}^*$ with oracle $g:\N\to\Q^n$, does the following. If $\pi$ is of the form $\pi=\pi_1\pi_2$, where $U(\pi_1,\lambda)=t\in\N$, then $M$ simulates $U(\pi_2,g(t))$. Let $c$ be an optimality constant for the oracle Turing machine $M$.
				
				To see that $c$ affirms the lemma, let $x\in\R^m$, $y\in\R^n$, and $r,s\in\N$. Let $q=y\upharpoonright (s+\log\sqrt{n})$, the truncation of the binary expansions of each of $y$'s coordinates to $s+\log\sqrt{n}$ bits to the right of the binary point. Let $\pi_s\in\{0,1\}^*$ testify to the value of $K(s)$, and let $\pi_x$ testify to the value of $\hat{K}_r(x|q)$. Then
				\[q\in\Q^n\cap B_{2^{-s}}(y)\]
				and
				\[M^{y}(\pi_s\pi_x)=U(\pi_x,q)\in\Q^m\cap B_{2^{-r}}(x)\,,\]
				so
				\begin{align*}
					K_r^y(x)&\leq K_{M,r}^y(x)+c\\
					&\leq |\pi_s\pi_x|+c\\
					&=\hat{K}_r(x|q)+K(s)+c\\
					&\leq K_{r,s}(x|y)+K(s)+c\,.
				\end{align*}
			\end{proof}
			\subsection{Claims in proof of Lemma~\ref{lem:main}}
			{
				\renewcommand{\thecla}{\ref{cla:floor}}
				\begin{cla}
					For every $r\in\N$, $\hat{K}_r(b|m)=K_r(b|m_r)+o(r)$, where $m_r=2^{-r}\lfloor m\cdot2^r\rfloor$.
				\end{cla}
				\addtocounter{cla}{-1}
			}
			\begin{proof}
				$K_r(b|m)\geq \hat{K}_r(b|m_r)$ by definition, since $m_r\in B_{2^{-r}}(m)$.
				
				Let $\hat{b}\in B_{2^{-r}}(b)$ be such that $K(\hat{b}|m_r)=\hat{K}_r(b|m_r)$. Then
				\[|(\hat{b},m_r)-(b,m)|\leq\sqrt{2}\cdot2^{-r}<2^{1-r}\,,\]
				so
				\[K(\hat{b},m_r)\geq K_{r-1}(b,m)=K_r(b,m)+o(r)\,,\]
				by Corollary 3.9 of~\cite{CasLut15}.
				
				Let $\mu$ testify to the value of $K_r(m)$, and let $\hat{m}=U(\mu)$. Then $|\hat{m}-m|<2^{-r}$, so $|\hat{m}-m_r|<2^{1-r}$. Thus once $\hat{m}$ and $r$ have been specified, there are at most four possible values for $m_r$. Therefore there is a self-delimiting Turing machine that takes as input $\mu$, an encoding of $r$ of length $o(r)$, and $O(1)$ additional bits and outputs $m_r$. We conclude that $K(m_r)\leq K_r(m)+o(r)$. Therefore we have
				\begin{align*}
				\hat{K}_r(b|m_r)&=K(\hat{b}|m_r)\\
				&=K(\hat{b},m_r)-K(m_r)+o(r)\\
				&\geq K_r(b,m)+o(r)-(K_r(m)+o(r))+o(r)\\
				&=K_r(b|m)+o(r)\,,
				\end{align*}
				by Theorem~\ref{thm:KrChainRule}.
			\end{proof}
			{
				\renewcommand{\thecla}{\ref{cla:ExistsH}}
				\begin{cla}
					For each $x\in[0,1]$ and $r\in\mathbb{N}$, there exists an $h\in\mathbb{N}$ such that $M$ halts on input $(\rho\pi_r\sigma_r\eta)$ with  $M(\rho\pi_r\sigma_r\eta)\in B_{2^{1-r}}(m,b,x)$, where $U(\rho)=r$ and $U(\eta)=h$.
				\end{cla}
				\addtocounter{cla}{-1}
			}
			\begin{proof}
				Fix $x\in[0,1]$ and $r\in\mathbb{N}$. It is clear that for some $j\in\{0,1,\ldots,2^r\}$, $|u_j-m|<2^{-r}$. By the definition of $K_r(b|m)$, $u_j\in\mathbb{Q}\cap B_{2^{-r}}(m)$ implies that $U(\pi_r,u_j)$ halts and outputs $v_j\in\mathbb{Q}\cap B_{2^{-r}}(b)$. $U(\sigma_r)\in B_{2^{-r}}(x,mx+b)$ by the definition of $\sigma_r$, so $|p-x|<2^{-r}$. It follows that \[|(u_j,v_j,p)-(m,x,b)|<\sqrt{3(2^{-r})^2}<2^{1-r}\,.\]
				
				It remains to show that $|u_ip+v_j-q|<2^{2-r}$. To do so, we repeatedly apply the triangle inequality and use the fact that $x,m\in[0,1]$:
				\begin{align*}
				|u_ip+v_j-q|&\leq|u_ip+v_j-(mx+b)|+|mx+b-q|\\
				&<|u_jp-mx+v_j-b|+2^{-r}\\
				&\leq|u_jp-mx|+|v_j-b|+2^{-r}\\
				&<|u_jp-u_jx|+|u_jx-mx|+2^{1-r}\\
				&\leq|p-x|+|u_j-m|+2^{1-r}\\
				&<2^{2-r}\,.
				\end{align*}
			\end{proof}
			{
				\renewcommand{\thecla}{\ref{cla:Meas0}}
				\begin{cla}
					For almost every $x\in[0,1]$, $\log(h(x,r))=o(r)$.
				\end{cla}
				\addtocounter{cla}{-1}
			}
			\begin{proof}
				By the countable additivity of Lebesgue measure, it suffices to show for every $k\in\N$ that the set
				\[D_k=\left\{x\in[0,1]:\exists\textrm{ infinitely many }r\in\N\textrm{ such that }\log(h(x,r))>r/k\right\}\]
				has Lebesgue measure $0$. For each $r\in\mathbb{N}$, let $D_{k,r}=\{x:h(x,r)>2^{r/k}\}$. We now estimate $\lambda(D_{k,r})$, the Lebesgue measure of $D_{k,r}$.
				
				For fixed $x$ and $r$, the algorithm run by the Turing machine $M$ entails \[h(x,r)\leq\left|\left\{i:|u_ip+v_i-q|<2^{2-r}\right\}\right|\,.\]
				For fixed $i$,
				\begin{align*}
				|u_ip+v_i-q|&>|u_ix-u_ip|+|u_ip+v_i-q|-2^{-r}\\
				&\geq|u_ix+v_i-q|-2^{-r}\\
				&>|u_ix+v_i-q|+|q-(mx+b)|-2^{1-r}\\
				&\geq|u_ix+v_i-(mx+b)|-2^{1-r}\,.
				\end{align*}
				That is,
				\[\left\{i:|u_ip+v_i-q|<2^{2-r}\right\}\subseteq\left\{i:|u_ix+v_i-(mx+b)|-2^{1-r}<2^{2-r}\right\}\,,\]
				so
				\[h(x,r)\leq\left|\left\{i:|u_ix+v_i-(mx+b)|<2^{3-r}\right\}\right|\,.\]
				
				For fixed $r$ and $i=0,1,\ldots,2^r$, define
				\[C^r_i=\{x\in[0,1]:|u_ix+v_i-(mx+b)|<2^{3-r}\}\,,\]
				For each $i$, if $m=u_i$, then $C^r_i$ is either $[0,1]$ or empty; otherwise, $C^r_i$ is an interval of length
				\[\lambda(C^r_i)\leq\min\left\{\frac{2^{3-r}}{|u_i-m|},1\right\}\,.\]
				Notice that for each $k=0,\ldots,2^r$, there are at most 2 values of $i$ for which
				$2^{-r}k\leq|u_i-m|<2^{-r}(k+1)$,
				so we have
				\begin{align*}
				\int_0^1 h(x,r)dx&\leq\sum\limits_{i=0}^{2^r}\lambda(C^r_i)\\
				&\leq 2+\sum\limits_{k=1}^{2^r}2\frac{2^{3-r}}{2^{-r}k}\\
				&=2+2^4\sum\limits_{k=1}^{2^r}\frac{1}{k}\\
				&<r2^6\,.
				\end{align*}
				Thus, as $h(x,r)>2^{r/k}$ for all $x\in D_{k,r}$, 
				\[\lambda(D_{k,r})<\frac{r2^6}{2^{r/k}}=r2^{6-r/k}\,.\]
				This implies that
				\[\sum_{r=1}^\infty\lambda(D_{k,r})<\infty\,,\]
				so the Borel-Cantelli Lemma tells us that $\lambda(D_k)=0$.
			\end{proof}
			{
				\renewcommand{\thecla}{\ref{cla:pointdim}}
				\begin{cla}
					For every $x\in[0,1]$, if $\log(h(x,r))=o(r)$, then
					\[\liminf_{r\to\infty}\frac{K_r(m,b,x)-K_r(b|m)}{r}\leq\dim(x,mx+b)\,.\]
				\end{cla}
				\addtocounter{cla}{-1}
			}
			\begin{proof}
				For fixed $r$, Claim~\ref{cla:ExistsH} gives
				\[K_{r-1}(m,b,x)\leq K(u_i,v_i,p)\leq K_M(u_i,v_i,p)+c_M\,,\]
				where $c_M$ is an optimality constant for $M$. Let $\rho$ and $\eta$ testify to the values of $K(r)$ and $K(h(x,r))$, respectively. Then $K_M(u_i,v_i,p)\leq|\rho\pi_r\sigma_r\eta|$. By our choices of $\rho, \pi_r, \sigma_r$, and $\eta$,
				\begin{align*}
				|\rho\pi_r\sigma_r\eta|&=K(r)+\hat{K}_r(b|m_r)+K_r(x,mx+b)+K(h(x,r))\\
				&=K(r)+K_r(b|m)+K_r(x,mx+b)+K(h(x,r))+o(r)\,,
				\end{align*}
				by Claim~\ref{cla:floor}.
				By Corollary 3.9 of~\cite{CasLut15},
				\begin{align*}
				&\liminf_{r\to\infty}\frac{K_r(m,b,x)-K_r(b|m)}{r}\\
				&=\liminf_{r\to\infty}\frac{K_{r-1}(m,b,x)-K_r(b|m)+o(r)}{r}\\
				&\leq\liminf_{r\to\infty}\frac{K(r)+K_r(x,mx+b)+K(h(x,r))+o(r)}{r}\\
				&\leq\liminf_{r\to\infty}\frac{K_{r}(x,mx+b)}{r}+\limsup_{r\to\infty}\frac{K(r)+K(h(x,r))+o(r)}{r}\\
				&=\dim(x,mx+b)+\limsup_{r\to\infty}\frac{K(h(x,r))}{r}\,.
				\end{align*}
				Applying Observation~\ref{obs:KBoundN}, for some constant $c$,
				\begin{align*}
				\limsup_{r\to\infty}\frac{K(h(x,r))}{r}&\leq\limsup_{r\to\infty}\frac{\log(1+h(x,r))+2\log\log(2+h(x,r))+c}{r}\\
				&=\limsup_{r\to\infty}\frac{\log(h(x,r))+2\log\log(h(x,r))}{r}\,.
				\end{align*}
				If $\log(h(x,r))={o(r)}$, then this is
				\[\limsup_{r\to\infty}\frac{o(r)+2\log(o(r))}{r}=0\,.\]
			\end{proof}
		\subsection{Observations about Kolomogorov Complexity in\\Euclidean Space}
			\begin{obs}\label{obs:BallLattice}
				For every open ball $B\subseteq\R^m$ of radius $2^{-r}$,
				\[B\cap2^{-\left(r+\left\lfloor\frac{1}{2}\log m\right\rfloor+1\right)}\Z^m\neq\emptyset\,.\]
			\end{obs}
			
			For $a\in\Z^m$, let $|a|$ denote the distance from the origin to $a$.
			\begin{obs}\label{obs:KBoundN}
				There is a constant $c_0\in\N$ such that, for all $j\in\N$,
				\[K(j)\leq\log(1+j)+2\log\log(2+j)+c_0\,.\]
			\end{obs}
			\begin{obs}\label{obs:KBoundZm}
				There is a constant $c\in\N$ such that, for all $a\in\Z^m$,
				\[K(a)\leq m\log(1+|a|)+\varepsilon(|a|)\,,\]
				where $\varepsilon(t)=c+2\log\log(2+t)$.
			\end{obs}
			
			Observation \ref{obs:KBoundN} holds by a routine technique \cite{LiVit08}. The proof of Observation \ref{obs:KBoundZm} is also routine:
			
			\begin{proof}
				Fix a computable, nonrepeating enumeration $a_0,a_1,a_2,\ldots$ of $\Z^m$ in which tuples $a_j$ appear in nondecreasing order of $|a_j|$. Let $M$ be a Turing machine such that, for all $\pi\in\{0,1\}^*$, if $U(\pi)\in\N$, then $M(\pi)=a_{U(\pi)}$. Let $c=c_0+c_M+m+\lceil2\log m\rceil+2$, where $c_0$ is as in Observation~\ref{obs:KBoundN} and $c_M$ is an optimality constant for $M$.
				
				To see that $c$ affirms Observation~\ref{obs:KBoundZm}, let $a\in\Z^m$. Let $j\in\N$ be the index for which $a_j=a$, and let $\pi\in\{0,1\}^*$ testify to the value of $K(j)$. Then $M(\pi)=a_{U(\pi)}=a_j=a$, so
				\[K(a)\leq K_M(a)+c_M\leq |\pi|+c_M=K(j)+c_M\,.\]
				It follows by Observation~\ref{obs:KBoundN} that
				\begin{equation}\label{eq:A}
				K(a)\leq\log(1+j)+2\log\log(2+j)+c+c_M\,.
				\end{equation}
				We thus estimate $j$.
				
				Let $B$ be the closed ball of radius $|a|$ centered at the origin in $\Z^m$, and let $Q$ be the solid, axis-parallel $m$-cube circumscribed about $B$. Let $B^\prime=B\cap\Z^M$ and $Q^\prime=Q\cap\Z^m$. Then
				\[j\leq|B^\prime|-1\leq|Q^\prime|-1\leq(2|a|+1)^m-1\,,\]
				so (\ref{eq:A}) tells us that
				\begin{align*}
				K(a)&\leq m\log(2|a|+1)+2\log\log(1+(2|a|+1)^m)+c+c_M\\
				&\leq m\log(2|a|+2)+2\log(m\log(2|a|+4))+c+c_M\,.
				\end{align*}
				Since
				\[m\log(2|a|+2)=m+m\log(1+|a|)\]
				and
				\begin{align*}
				\log(m\log(2|a|+4))&=\log m+\log(1+\log(2+|a|))\\
				&\leq\log m+1+\log\log(2+|a|)\,,
				\end{align*}
				it follows that $K(a)\leq m\log(1+|a|)+\varepsilon(|a|)$.
			\end{proof}
			\begin{obs}\label{obs:addrat}
				For every $r,n\in\N$, $x\in\R^n$, and $q\in\Q^n$, \[K_r(x+q)=K_r(x)+O(1)\,.\]
			\end{obs}
			\begin{proof}
				Let $M$ be a self-delimiting Turing machine such that $M(\pi\kappa)=U(\pi)+U(\kappa)$ whenever $U(\pi),U(\kappa)\in\Q^n$. If $\pi$ is a witness to $K_r(x)$ and $\kappa$ is a witness to $q$, then $M(\pi\kappa)=p+q$ for some $p\in B_{2^{-r}}(x)$, so $M(\pi\kappa)\in B_{2^{-r}}(x+q)$. Thus
				\[K_r(x+q)\leq K_r(x)+K(q)+c\,,\]
				where $c$ is a machine constant for $M$. Since $K(q)$ is constant in $r$, we have $K_r(x+q)\leq K_r(x)+O(1)$. Applying the same argument with $-q$ replacing $q$ completes the proof.
			\end{proof}

			\subsection{Linear Sensitivity of $\hat{K}_r(x|q)$ to $r$}
			\begin{lem}\label{lem:LinSensCondKr}
				There is a constant $c_1\in\N$ such that, for all $x\in\R^m$, $q\in\Q^n$, and $r,\Delta r\in\N$,
				\[\hat{K}_r(x|q)\leq \hat{K}_{r+\Delta r}(x|q)\leq \hat{K}_r(x|q)+m\Delta r+\varepsilon_1(r,\Delta r)\,,\]
				where $\varepsilon_1(r,\Delta r)=2\log(1+\Delta r)+K(r,\Delta r)+c_1$.
			\end{lem}
			\begin{proof}
				Let $M$ be a Turing machine such that, for all $\pi_1,\pi_2,\pi_3\in\{0,1\}^*$ and $q\in\Q^n$, if $U(\pi_1,q)=p\in\Q^m$, $U(\pi_2)=(r,\Delta r)\in\N^2$, and $U(\pi_3)=a\in\Z^m$, then $M(\pi_1\pi_2\pi_3,q)=p+2^{-r^*}a$, where $r^*=r+\Delta r+\left\lfloor\tfrac{1}{2}\log m\right\rfloor+1$. Let $c_1=c+c_M+3m+m\left\lfloor\tfrac{1}{2}\log m\right\rfloor+\left\lceil 2\log(3+\left\lfloor\tfrac{1}{2}\log m\right\rfloor)\right\rceil$, where $c$ is the constant from Observation \ref{obs:KBoundZm} and $c_M$ is an optimality constant for $M$.
				
				To see that $c_1$ affirms the lemma, let $x$, $q$, $r$, and $\Delta r$ be as given. The first inequality holds trivially. To see that the second inequality holds, let $\pi_1,\pi_2\in\{0,1\}^*$ testify to the values of $\hat{K}_r(x|q)$ and $K(r,\Delta r)$, respectively. Let $B=B_{2^{-r}}(x)$, $B^\prime=B_{2^{-(r+\Delta r)}}(x)$ and $p=U(\pi_1,q)$, noting that $p\in\Q^m\cap B$. Applying Observation \ref{obs:BallLattice} to the ball $B^\prime-p$ tells us that
				\[(B^\prime-p)\cap2^{-r^*}\Z^m\neq\emptyset\,,\]
				i.e., that
				\[B^\prime\cap(p+2^{-r^*}\Z^m)\neq\emptyset\,.\]
				So fix a point $p^\prime\in B^\prime\cap(p+2^{-r^*}\Z^m)$, say, $p^\prime=p+2^{-r^*}a$, where $a\in\Z^m$, and let $\pi_3\in\{0,1\}^*$ testify to the value of $K(a)$. Then
				\[M(\pi_1\pi_2\pi_3,q)=p^\prime\in\Q\cap B^\prime\,,\]
				so
				\begin{align*}
				\hat{K}_{r+\Delta r}(x|q)&\leq K(p^\prime|q)\\
				&\leq \hat{K}_M(p^\prime|q)+c_M\\
				&\leq |\pi_1\pi_2\pi_3|+c_M\,.
				\end{align*}
				By our choice of $\pi_1$, $\pi_2$, and $\pi_3$, this implies that
				\begin{equation}\label{eq:B}
				\hat{K}_{r+\Delta r}(x|q)\leq \hat{K}_r(x|q)+K(r,\Delta r)+K(a)+c_M\,.
				\end{equation}
				We thus estimate $K(a)$.
				
				Since
				\begin{align*}
				|a|&= 2^{r^*}|p^\prime-p|\\
				&\leq 2^{r^*}(|p^\prime-x|+|p-x|)\\
				&< 2^{r^*}\left(2^{-(r+\Delta r)}+2^{-r}\right)\\
				&= 2^{1+\left\lfloor\frac{1}{2}\log m\right\rfloor}\left(1+2^{\Delta r}\right)\,,
				\end{align*}
				Observation \ref{obs:KBoundZm} tells us that
				\begin{align*}
				K(a)&\leq m\log\left(1+2^{\left\lfloor\frac{1}{2}\log m\right\rfloor}\left(1+2^{\Delta r}\right)\right)+\varepsilon(|a|)\\
				&\leq m\log\left(2^{\Delta r+3+\left\lfloor\frac{1}{2}\log m\right\rfloor}\right)+\varepsilon(|a|)\,,
				\end{align*}
				i.e., that
				\begin{equation}\label{eq:C}
				K(a)\leq m\Delta r+3m+m\left\lfloor\tfrac{1}{2}\log m\right\rfloor+\varepsilon(|a|)\,,
				\end{equation}
				where
				\begin{align*}
				\varepsilon(|a|)&\leq c+2\log\log\left(2+2^{1+\left\lfloor\frac{1}{2}\log m\right\rfloor}\left(1+2^{\Delta r}\right)\right)\\
				&\leq c+2\log\log\left(2^{\Delta r+3+\left\lfloor\frac{1}{2}\log m\right\rfloor}\right)\\
				&= c+2\log\left(\Delta r+\left\lfloor\tfrac{1}{2}\log m\right\rfloor+3\right)\\
				&\leq c+2\log\left(\left(1+\Delta r\right)\left(3+\left\lfloor\tfrac{1}{2}\log m\right\rfloor\right)\right)\\
				&= c+2\log(1+\Delta r)+2\log\left(3+\left\lfloor\tfrac{1}{2}\log m\right\rfloor\right)\,.
				\end{align*}
				It follows by (\ref{eq:B}) and (\ref{eq:C}) that
				\[\hat{K}_{r+\Delta r}(x|q)\leq \hat{K}_r(x|q)+m\Delta r+\varepsilon_1(r,\Delta r)\,.\]
			\end{proof}
			\subsection{Linear Sensitivity of $K_{r,s}(x|y)$ to $s$}
			\begin{lem}\label{lem:LinSensCondKrs}
				There is a constant $c_2\in\N$ such that, for all $x\in\R^m$, $y\in\R^n$, and $r,s,\Delta s\in\N$,
				\[K_{r,s}(x|y)\geq K_{r,s+\Delta s}(x|y)\geq K_{r,s}(x|y)-n\Delta s-\varepsilon_2(s,\Delta s)\,,\]
				where $\varepsilon_2(s,\Delta s)=2\log(1+\Delta s)+K(s,\Delta s)+c_2$.
			\end{lem}
			\begin{proof}
				Let $M$ be a Turing machine such that, for all $\pi_1,\pi_2,\pi_3\in\{0,1\}^*$ and $q\in\Q^n$, if $U(\pi_1)=(s,\Delta s)\in\N^2$ and $U(\pi_2)=a\in\Z^m$, then $M(\pi_1\pi_2\pi_3,q)=U(\pi_3,q+2^{-s^*}a)$, where $s^*=s+\Delta s+\left\lceil\tfrac{1}{2}\log n\right\rceil$. Let $c_2=c+c_M+3n+n\left\lfloor\tfrac{1}{2}\log n\right\rfloor+2\left\lceil2\log(3+\left\lfloor\tfrac{1}{2}\log n\right\rfloor)\right\rceil$, where $c$ is the constant from Observation \ref{obs:KBoundZm} and $c_M$ is an optimality constant for $M$.
				
				To see that $c_2$ affirms the lemma, let $x$, $y$, $r$, $s$, and $\Delta s$ be as given The first inequality holds trivially. To see that the second inequality holds, let $B=B_{2^-s}(y)$, $B^\prime=B_{2^{-(s+\Delta s)}}(y)$, and $q\in\Q^n\cap B$. It suffices to prove that
				\begin{equation}\label{eq:D}
				\hat{K}_r(x|q)\leq K_{r,s+\Delta s}(x|y)+n\Delta s+\varepsilon_2(s,\Delta s)\,.
				\end{equation}
				
				Let $\pi_1\in\{0,1\}^*$ testify to the value of $K(s,\Delta s)$. Applying Observation \ref{obs:BallLattice} to the ball $B^\prime-q$ tells us that
				\[(B^\prime-q)\cap2^{-s^*}\Z^n\neq\emptyset\,,\]
				i.e., that
				\[B^\prime\cap(q+2^{-s^*}\Z^n)\neq\emptyset\,.\]
				So fix a point $q^\prime\in B^\prime\cap(q+2^{-s^*}\Z^n)$, say, $q^\prime=q+2^{-s^*}a$, where $a\in\Z^n$. Note that
				\begin{equation}\label{eq:E}
				\hat{K}_r(x|q^\prime)\leq K_{r,s+\Delta s}(x|y).
				\end{equation}
				Let $\pi_2,\pi_3\in\{0,1\}^*$ testify to the values of $K(a)$ and $\hat{K}_r(x|q^\prime)$, respectively, noting that $U(\pi_3,q^\prime)=p$ for some $p\in\Q^m\cap B_{2^{-r}}(x)$. Then
				\[M(\pi_1\pi_2\pi_3,q)=U(\pi_3,q^\prime)=p\in\Q^m\cap B_{2^{-r}}(x)\,,\]
				so
				\begin{align*}
				\hat{K}_r(x|q)&\leq K(p|q)\\
				&\leq K_M(p|q)+c_M\\
				&\leq |\pi_1\pi_2\pi_3|+c_M\,.
				\end{align*}
				By our choice of $\pi_1$, $\pi_2$, and $\pi_3$, and by (\ref{eq:E}), this implies that
				\begin{equation}\label{eq:F}
				\hat{K}_r(x|q)\leq K_{r,s+\Delta s}(x|y)+K(a)+K(s,\Delta s)+c_M\,.
				\end{equation}
				We thus estimate K(a).
				
				Since
				\begin{align*}
				|a|&= 2^{s^*}|q^\prime-q|\\
				&\leq 2^{s^*}(|q^\prime-y|+|q-y|)\\
				&< 2^{s^*}(s^{-(s+\Delta s)}+2^{-s})\\
				&= 2^{1+\left\lfloor\frac{1}{2}\log n\right\rfloor}(1+2^{\Delta s})\,,
				\end{align*}
				Observation \ref{obs:KBoundZm} tells us that
				\begin{align*}
				K(a)&\leq n\log(1+2^{1+\left\lfloor\frac{1}{2}\log n\right\rfloor}(1+2^{\Delta s}))+\varepsilon(|a|)\\
				&\leq n\log(2^{\Delta s+3+\left\lfloor\frac{1}{2}\log n\right\rfloor})+\varepsilon(|a|)\,,
				\end{align*}
				i.e., that
				\begin{equation}\label{eq:G}
				K(a)\leq n\Delta s+3n+n\left\lfloor\tfrac{1}{2}\log n\right\rfloor+\varepsilon(|a|),
				\end{equation}
				where
				\begin{align*}
				\varepsilon(|a|)&\leq c+2\log\log(2+2^{1+\left\lfloor\frac{1}{2}\log n\right\rfloor}(1+2^{\Delta s}))\\
				&\leq c+2\log\log(2^{\Delta s+3+\left\lfloor\frac{1}{2}\log n\right\rfloor})\\
				&= c+2\log(\Delta s+3+\left\lfloor\tfrac{1}{2}\log n\right\rfloor)\\
				&\leq c+2\log\big((1+\Delta s)(3+\left\lfloor\tfrac{1}{2}\log n\right\rfloor)\big)\\
				&= c+2\log(1+\Delta s)+2\log(3+\left\lfloor\tfrac{1}{2}\log n\right\rfloor)\,.
				\end{align*}
				It follows by (\ref{eq:F}) and (\ref{eq:G}) that (\ref{eq:D}) holds.
			\end{proof}
		\subsection{Proof of Theorem~\ref{thm:CondDimRobust}}
			{
				\renewcommand{\thethm}{\ref{thm:CondDimRobust}}
				\begin{thm}
					Let $s:\N\to\N$. If $|s(r)-r|=o(r)$, then, for all $x\in\R^m$ and $y\in\R^n$,
					\[\dim(x|y)=\liminf_{r\to\infty}\frac{K_{r,s(r)}(x|y)}{r}\,,\]
					and
					\[\Dim(x|y)=\limsup_{r\to\infty}\frac{K_{r,s(r)}(x|y)}{r}\,.\]
				\end{thm}
				\addtocounter{thm}{-1}
			}
			\begin{proof}
				Assume the hypothesis. Define $s^-,s^+:\N\to\N$ by
				\[s^-(r)=\min\{r,s(r)\},\;s^+(r)=\max\{r,s(r)\}\,.\]
				Lemma \ref{lem:LinSensCondKrs} tells us that, for all $x\in\R^m$ and $y\in\R^n$,
				\begin{align*}
				K_{r,s^-(r)}(x|y)&\geq K_{r,r}(x|y)\\
				&\geq K_{r,s^+(r)}(x|y)\\
				&\geq K_{r,s^-(r)}(x|y)-O(s^+(r)-s^-(r))-o(r)\\
				&= K_{r,s^-(r)}(x|y)-O(|s(r)-r|)-o(r)\\
				&= K_{r,s^-(r)}(x|y)-o(r)\,.
				\end{align*}
				Since
				\[K_{r,s^-(r)}(x|y)\geq K_{r,s(r)}(x|y)\geq K_{r,s^+(r)}(x|y)\,,\]
				it follows that
				\[\big|K_{r,s(r)}(x|y)-K_{r,r}(x|y)\big|=o(r)\,.\]
				The theorem follows immediately.
			\end{proof}
			\subsection{Proof of Lemma~\ref{lem:MutualInfo}}
			{
				\renewcommand{\thethm}{\ref{lem:MutualInfo}}
				\begin{lem}
					For all $x\in\R^m$, $y\in\R^n$, and $r\in\N$,
					\[I_r(x:y)=K_r(x)-K_r(x|y)+o(r)\,.\]
				\end{lem}
				\addtocounter{thm}{-1}
			}
			\begin{proof}
				Let $B_x=B_{2^{-r}}(x)\cap\Q^m$ and $B_y= B_{2^{-r}}(y)\cap\Q^n$. Let $p_0$ and $q_0$ be $K$-minimizers for $B_x$ and $B_y$, respectively, such that
				\begin{equation}\label{eq:MI1}
				I_r(x:y)=I(p_0:q_0)+o(r)\,.
				\end{equation}
				These exist by Theorem 4.6 of~\cite{CasLut15}. Then
				\begin{align*}
				K_r(x)-K_r(x|y)&=\min_{p\in B_x}K(p)-\max_{q\in B_y}\min_{p\in B_x}K(p|q)\\
				&\geq\min_{p\in B_x}K(p)-\min_{p\in B_x}\max_{q\in B_y}K(p|q)\\
				&=\min_{p\in B_x}K(p)-\min_{p\in B_x}K(p|q_0)+o(r)\,,
				\intertext{by Lemma 4.2 and Observation 3.7 of~\cite{CasLut15}.}
				&=K(p_0)-\min_{p\in B_x}K(p|q_0)+o(r)\\
				&\geq K(p_0)-K(p_0|q_0)+o(r)\\
				&=I(p_0:q_0)+o(r)\\
				&=I_r(x:y)+o(r)\,.
				\end{align*}
				
				For the other direction, let $p_1\in B_x$ be such that
				\[K(p_1|q_0)=\min_{p\in B_x}K(p|q_0)\,.\]
				By Lemma 4.5 of~\cite{CasLut15},
				\begin{align*}
				I(p_0:q_0)&\geq K(p_1)-K(p_1|p_0,K(p_0))-K(p_1|q_0,K(q_0))+o(r)\\
				&\geq K(p_1)-K(p_1|p_0,K(p_0))-K(p_1|q_0)+o(r)\numberthis\label{eq:MI2}\,.
				\end{align*}
				Now
				\begin{align*}
				K(p_0)+K(p_1|p_0,K(p_0))+o(r)&=K(p_0,p_1)\\
				&=K(p_1)+K(p_0|p_1,K(p_1))+o(r)\\
				&\leq K(p_1)+K(p_0|p_1)+o(r)\\
				&=K(p_1)+o(r)\,,
				\end{align*}
				by Corollary 4.4 of~\cite{CasLut15}. So
				\[K(p_1)-K(p_1|p_0,K(p_0))\geq K(p_o)+o(r)\,,\]
				thus by (\ref{eq:MI2}),
				\begin{align*}
				I(p_0:q_0)&\geq K(p_0)-K(p_1|q_0)+o(r)\\
				&=K(p_0)-\min_{p\in B_x}K(p|q_0)+o(r)\\
				&=K_r(x)-\min_{p\in B_x}K(p|q_0)+o(r)\\
				&\geq K_r(x)-\max_{q\in B_y}\min_{p\in B_x}K(p|q)+o(r)\\
				&=K_r(x)-K_r(x|y)+o(r)\,.
				\end{align*}
				Then by (\ref{eq:MI1}), $I_r(x:y)\geq K_r(x)-K_r(x|y)+o(r)$, so equality holds.
			\end{proof}
\end{document}